\documentclass[letterpaper,12pt]{amsart}
\usepackage{latexsym,array,delarray,amsthm,amssymb,epsfig}
\usepackage{color}

\usepackage{caption,enumerate}
\usepackage{subcaption}
\usepackage{tikzsymbols}

\theoremstyle{plain}
\newtheorem{thm}{Theorem}[section]
\newtheorem{lemma}[thm]{Lemma}
\newtheorem{op}[thm]{Proposition}
\newtheorem{cor}[thm]{Corollary}
\newtheorem{conj}[thm]{Conjecture}
\newtheorem*{thm*}{Theorem}
\newtheorem*{lemma*}{Lemma}
\newtheorem*{prop*}{Proposition}
\newtheorem*{cor*}{Corollary}
\newtheorem*{conj*}{Conjecture}

\theoremstyle{definition}
\newtheorem{defn}[thm]{Definition}
\newtheorem{ex}[thm]{Example}

\newtheorem{ques}[thm]{Question}

\theoremstyle{remark}
\newtheorem*{rmk}{Remark}

\newcommand{\rr}{\mathbb{R}}
\newcommand{\cc}{\mathbb{C}}

\newcommand{\ind}{\mbox{$\perp \kern-5.5pt \perp$}}

\begin{document}

\title{Distinguishing Phylogenetic Networks}
\author{Elizabeth Gross and Colby Long}

\email{elizabeth.gross@sjsu.edu}
\email{long.1579@mbi.osu.edu} 
\address{Department of Mathematics and Statistics, One Washington Square,  San Jos\'{e} State University, San Jos\'{e}, CA, 95192-0103, USA}
 \address{Mathematical Biosciences Institute, The Ohio State University, 1735 Neil Ave., Columbus OH, 43210, USA}

\begin{abstract}  Phylogenetic networks are becoming increasingly popular in phylogenetics since they have the ability to describe a wider range of evolutionary events than their tree counterparts.  In this paper, we study Markov models on phylogenetic networks and their associated geometry. We restrict our attention to large-cycle networks, networks with a single undirected cycle of length at least four. Using tools from computational algebraic geometry, we show that the semi-directed network topology is generically identifiable for Jukes-Cantor large-cycle network models.   
\end{abstract}

\maketitle

\section{Introduction}
\label{sec: introduction}

There are many reasons why a single phylogenetic tree may fail to fully describe the evolutionary history of a group of taxa. Issues such as hybridization, horizontal gene transfer, and incomplete lineage sorting are known to cause discordance between gene trees \cite{maddison1997,pamilonei1988,syvanen1994}.
To account for hybridization and horizontal gene transfer, \emph{evolutionary phylogenetic networks} have recently come to the foreground of phylogenetics.  These networks model the process of evolution along a directed acyclic graph where certain edges in the network represent reticulation events. Since the network topology is meant to reflect the actual history of a group of taxa, the topologies are often constrained to a class of networks considered biologically reasonable. Because of their increasing importance, many concepts from modeling and inference on phylogenetic trees are now being applied to phylogenetic networks. 
For example, a number of authors have incorporated the coalescent process into network models to account for incomplete lineage sorting \cite{Chifman2016, Solis2016, Yu2011}. There have also been a number of papers exploring network inference  \cite{gauthierlapointe2007, jinetal2006, Kubatko2009} and the combinatorial properties of different classes of phylogenetic networks \cite{Francis2016,Huber2015,Semple2016}. 
In fact, the body of work in this area has grown to the point that there are now several surveys on the topic (e.g., \cite{Moret2004,Huson2010b,Nakhleh2015}). 
Despite this, the present paper is one of the first to analyze
the algebraic and geometric properties of phylogenetic networks. 

The implicit goal underlying much of this work is to eventually be able to infer
the phylogenetic network that explains the evolutionary history of a group of taxa from biological data.
 Thus, a fundamental question about any phylogenetic network model is the identifiability of the underlying network, that is, whether or not the network topology can be uniquely identified from data generated by the network. Some positive results in this direction have been proven, for example, in \cite{Solis2016} it is shown that there are $4$-leaf networks that can be uniquely identified from the quartet topology distribution induced by the coalescent process. However, there are other results that should give pause to those attempting to reconstruct networks. For example, it is shown in \cite{Moret2004} that two distinct networks can share the same set of subtrees. Likewise, in \cite{Pardi2015} the authors show that two topologically distinct weighted networks can share the same set of weighted subtrees.

In this paper, we will consider Jukes-Cantor network models where the process of DNA sequence evolution is modeled as a Markov process proceeding along an $n$-leaf directed acyclic graph (DAG).
We are particularly interested in
the identifiability (or lack thereof) of the network topology from the distribution on $n$-tuples of DNA bases generated by the network.
This is distinct from the notion of identifiability discussed in 
\cite{Pardi2015}, as we do not assume any knowledge about which sites were produced by 
the same subtree of the network.  The two-state Cavender-Farris-Neyman model may seem the more natural starting point for 
our exploration of network identifiabilty. However, as is evident from our computations in 
Proposition \ref{prop: CFN}, the restricted coordinate space
for this model makes it impossible to identify small networks from one another, our main
strategy for eventually proving identifiability in the Jukes-Cantor case. 

 Since the Jukes-Cantor model is time-reversible, the precise location of the root within the network will be unidentifiable from the distribution.  However, we cannot simply study the unrooted topology of networks without orientation, since \emph{reticulation edges}, edges directed into vertices of indegree two, play a special role defining the distribution. 
Thus, our results concern the identifiability of the 
\emph{semi-directed network topology}, the unrooted, undirected network with distinguished reticulation edges. We will also restrict our attention to networks with only a single \emph{reticulation vertex} which we call \emph{cycle-networks}. We will refer to the set of all cycle-networks with cycle length greater than 4 as the class of \emph{large-cycle networks}.
The main result of this paper is the following theorem. 

\begin{thm}  
\label{thm: main}
The semi-directed network topology parameter of large-cycle Jukes-Cantor network models is generically identifiable. 
\end{thm}

Markov models on networks with a single reticulation vertex are very closely 
related to 2-tree mixture models but with some subtle differences that we discuss in Section \ref{sec: obtaindistribution}. 
Using techniques from algebraic statistics, it is shown in \cite{Allman} that the tree parameters of a 2-tree Jukes-Cantor mixture are generically identifiable. Here we adopt a similar approach. We associate to each network $\mathcal{N}$ an algebraic variety $\mathcal V_\mathcal{N}$ that is the Zariski closure of the set of probability distributions attained by varying the numerical parameters in the model on $\mathcal{N}$. We then study the associated ideals of the networks to find algebraic invariants that distinguish networks from one another. The two networks in Figure \ref{fig: 2IndistinguishableNetworks} demonstrate that the generic identifiability results for 2-tree mixtures do not apply 
for phylogenetic networks. These networks have different semi-directed network topologies and induce different
multisets of embedded trees. Suprisingly, however,  the algebraic variety for the network on the left is 
properly contained in that of the network on the right. 
This example highlights another contrast between cycle-networks and 2-tree mixtures, 
as the varieties for distinct $n$-leaf cycle-networks need not even be the same dimension.

\begin{figure}
\caption{Two phylogenetic networks for which $\mathcal V_{\mathcal{N}_1} \subseteq \mathcal V_{\mathcal{N}_2}$.}
\label{fig: 2IndistinguishableNetworks}
\includegraphics[width=10cm]{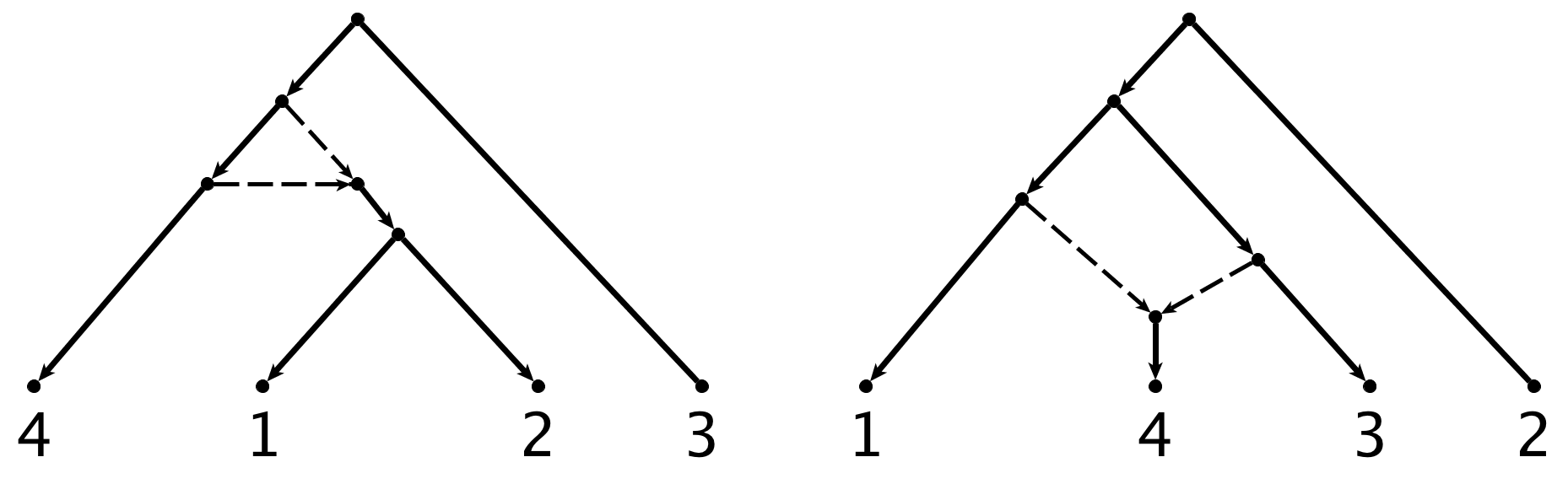}
\end{figure}

This paper is organized as follows. In Section \ref{sec:phylogenetic-networks}, we introduce the appropriate network terminology and rigorously define Jukes-Cantor network models to show how to obtain a probability distribution on DNA site patterns from a network. In Section \ref{sec: algstats}, we introduce the concept of generic identifiability and the algebraic background necessary to
prove the main results. Finally, in Section \ref{sec: results}, we present the main results about the identifiability of the semi-directed network topology. As Figure \ref{fig: 2IndistinguishableNetworks} illustrates, the network topologies are not generically identifiable, but we will be able to restrict to a class of networks that preserves identifiability. Additionally, we will be able to show many specific instances where identifiability fails. In Section \ref{sec: discussion}, we conclude with a discussion about the consequences of these results for inferring phylogenies.

\section{Phylogenetic Networks}
\label{sec:phylogenetic-networks}

The following network notation and terminology is 
adapted from \cite{Francis2016, FS15, Semple2016}.

\begin{defn}
A phylogenetic network  
$\mathcal{N}$
 on $X$ is a rooted acyclic digraph with no edges in parallel and satisfying the following properties:
\begin{enumerate}[(i)]
\item the root has out-degree two;
\item a vertex with out-degree zero has in-degree one, and the set of vertices with out-degree zero is $X$;
\item all other vertices either have in-degree one and out-degree two, or in-degree two and out-degree one.
\end{enumerate}
\end{defn}

Note that these are sometimes also referred to as 
\emph{binary phylogenetic networks}.
A vertex with indegree one and outdegree two is called a \emph{tree vertex} and a vertex with
 indegree two and outdegree one is called a 
 \emph{reticulation vertex} or simply a {reticulation}. 
 Edges directed into a reticulation edge are called \emph{reticulation edges} and all other edges are called \emph{tree edges}.
 Recall that the reason for introducing phylogenetic networks
is to model possible hybridization events and horizontal gene transfer. These events of course can only occur when two species
coexist in time. Considering only directed acyclic graphs precludes any paradoxes wherein genetic information is transported back in time. However, as noted in the introduction, due to the time-reversibility of the Jukes-Cantor model, we will not actually be able to identify the location of the root in the network from the models we define. Therefore, in this paper we are primarily interested in recovering the underlying \emph{semi-directed} network topology of a phylogenetic network. The semi-directed network topology is obtained from a phylogenetic network by suppressing the root node and undirecting all tree edges while the reticulation edges remain directed.

The class of phylogenetic networks is quite large and the algebraic approach we adopt becomes increasingly complicated as the number of reticulations in the network increases. Therefore, we begin here by studying networks that contain at most one reticulation vertex. Such structures are necessarily \emph{level-1} networks \cite{CJS04}, networks in which every edge belongs to at most one cycle. In fact, networks with exactly one reticulation edge contain a single undirected cycle, which motivates the following definition.

\begin{defn}\label{def:k-cycle} A \emph{cycle-network} is a semi-directed network with one reticulation vertex.  A {\emph k-cycle network} is a cycle-network with cycle size $k$. 
\end{defn}

Note that the cycle of a cycle-network always contains the reticulation vertex and the two reticulation edges. 
We will refer to an internal vertex contained in the cycle of 
a cycle-network as a \emph{cycle vertex}. For subsequent sections, it will be helpful to establish some conventions for $n$-leaf $k$-cycle networks. We can view an $n$-leaf $k$-cycle network $\mathcal N$ as a $k$-cycle with a tree $\mathcal T_v$ affixed to each cycle vertex $v$. 
Let $A_v$ be the leaf label set of $\mathcal{T}_v$. 
Then the cycle vertices of $\mathcal N$ induce an ordered partition of $[n]$. Label the reticulation vertex by $v_1$ and label the remaining cycle vertices $v_2, \ldots, v_k$ in a clockwise fashion so that $\min A_{v_2} < \min A_{v_k}$.  Using the shorthand $A_i$ for $A_{v_i}$, $\mathcal{N}$ induces the ordered partition $A_1 | A_2 | \ldots | A_k$.
As an example, Figure \ref{fig: kcycle network} depicts a 5-cycle network and the caption 
describes the ordered partition that it induces.
We call the unique $k$-leaf $k$-cycle network topology the \emph{$k$-sunlet}.

In Section \ref{sec: largekcycles}, we will be concerned with $k$-cycle networks with $k \geq 4$.  This motivates the following definition.

\begin{defn}\label{def:large-cycle} The set of \emph{large-cycle networks} is the collection of all $k$-cycle networks with $k \geq 4$. 
\end{defn}
\noindent The set of large-cycle networks is the focus of our main theorem, Theorem \ref{thm: main}.

\begin{figure}
\caption{A 5-cycle network that induces the ordered partition 
$A_1 = \{ 4,5\},
A_2 = \{ 2,6, 9\},
A_3 = \{ 8 \},
A_4 = \{1 \},$ and
$A_5 = \{3,7\}$.}
\label{fig: kcycle network}
\includegraphics[width=5cm]{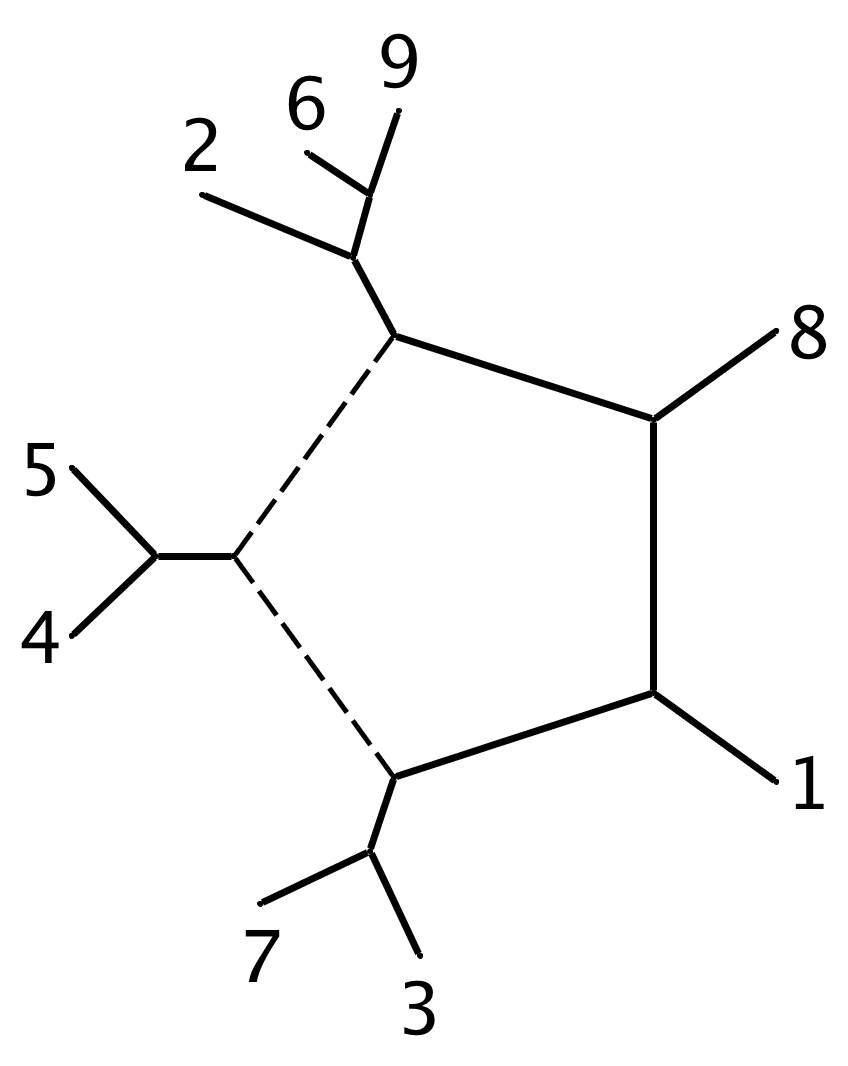}
\end{figure}

\subsection{Obtaining a distribution from a phylogenetic network}
\label{sec: obtaindistribution}

 In this section we describe
how to obtain a probability distribution on $n$-tuples of DNA bases
from an $n$-leaf phylogenetic network.
 As the network models we wish to discuss are a generalization of the nucleotide substitution models on phylogenetic trees, we begin by briefly reviewing these models. 

\subsubsection{Markov processes on phylogenetic trees}
\label{sec: tree models}

A phylogenetic model is a statistical model of molecular sequence evolution for a collection of $n$ taxa at a single DNA site. The tree parameter of such a model is an $n$-leaf rooted leaf-labeled tree $\mathcal{T}$ where the leaf vertices are labeled by the $n$ taxa. The internal nodes of the tree represent ancestors of the taxa at the leaves.
We denote the root of the tree by 
$\rho$ and associate to each node $v$ of $\mathcal{T}$ a random variable $X_v$ with state space $\{A,C,G,T\}$, corresponding to the four DNA bases.  The state of the random variable $X_v$ is meant to indicate the DNA base at the particular site being modeled in the taxon at $v$. 

Let $\boldsymbol{\pi}= (\pi_A,\pi_C,\pi_G,\pi_T) \in \mathbb{R}^4$ be the root distribution with $\pi_i = P(X_\rho = i)$, and  associate to each edge $e=uv$ of $\mathcal T$ a $4 \times 4$ transition matrix $M^e$ where the rows and columns are indexed by the elements of the state space. Assuming $u$ is the vertex closer to the root, $M^e_{ij}$ is equal to the conditional probability 
$P(X_v = j | X_u = i)$. The entries of the transition matrices are called the \emph{stochastic parameters} of the model. For a particular choice of parameters, the model returns a probability distribution on the set of $n$-tuples of DNA bases
that may be observed at the leaves of $\mathcal{T}$. To compute this distribution, we first consider an assignment of states $\phi \in \{A,C,G,T\}^{ V(\mathcal T)}$ to the vertices of $\mathcal{T}$ where $\phi(v)$ is the state of $X_v$. Then the probability of observing the state $\phi$ can be computed using the root distribution and the transition matrices. Specifically, letting $\Sigma(\mathcal{T})$
 be the set of edges of $\mathcal{T}$, this probability is equal to 

$$\displaystyle \prod_{e=uv \in \Sigma(\mathcal{T})} \
\pi_{\phi(\rho)}
M^e_{\phi(u),\phi(v)}.$$
Notice that this is a monomial in the stochastic parameters of the model. The probability of observing a particular state at the leaves can be obtained by marginalization, i.e. summing over all possible states of the internal nodes.
Therefore, the distribution on all $n$-tuples of possible leaf states is given by a polynomial map from the stochastic parameter space $\Theta _\mathcal{T}$ to the probability simplex 
$$\psi_\mathcal{T}: \Theta _\mathcal{T} \to \Delta^{4^n - 1}.$$

The model described above is referred to as the general Markov model; other specific phylogenetic models can be obtained by restricting the stochastic parameters.  For example, for the Jukes-Cantor model, all transition matrices are assumed to be of the form pictured in Figure \ref{fig: jcmatrix}. Because the rows of this matrix must sum to one, there is essentially 
a single parameter for each edge. 

Once a particular 
substitution model and a tree are specified, the image of the map $\psi_\mathcal{T}$ is called
the \emph{model associated to $\mathcal{T}$}, denoted 
$\mathcal{M}_\mathcal{T}$. The fact that $\psi_\mathcal{T}$ is a polynomial map makes the model $\mathcal{M}_\mathcal{T}$ amenable to study with algebraic geometry.

\begin{figure}
\caption{A transition matrix for the Jukes-Cantor model.}
\label{fig: jcmatrix}
$M^{e} = \begin{pmatrix}
     \alpha & \beta  & \beta  & \beta  \\
     \beta  & \alpha & \beta  & \beta  \\
     \beta  & \beta  & \alpha & \beta  \\
     \beta  & \beta  & \beta  & \alpha
    \end{pmatrix}$
\end{figure}

\subsubsection{Markov processes on phylogenetic networks}

Here we describe how to obtain a distribution on $n$-tuples of DNA bases
from a phylogenetic network by taking a convex combination of the distributions from 
phylogenetic tree models. 
These network models are also described in
\cite[\S 3.3]{Nakhleh2011}.  For this exposition, we assume that the network $\mathcal{N}$ is a \emph{tree-child} network \cite{Cardona2007}, that is, we assume that the child of a reticulation vertex is always a tree vertex.

Let $\mathcal{N}$ be an $n$-leaf phylogenetic network and 
associate a $4\times 4$ 
transition matrix from a nucleotide substitution model
to each edge of $\mathcal{N}$. 
Suppose $\mathcal{N}$ has $m$ reticulation vertices $w_1, \ldots, w_m$. 
Since each $w_i$ has indegree two, 
there are two edges, $e^0_i$ and $e^1_i$, 
directed into $w_i$. For $1 \leq i \leq m$, 
independently delete $e^0_i$ with probability
$\delta_i \in [0,1]$, otherwise, delete $e^1_i$.
Intuitively, the parameter $\delta_i$ corresponds to the probability that a particular site was inherited along edge $e^0_i$.
Encode this set of choices with a binary vector $\sigma \in \{0,1\}^m$ where a $0$ in the $i$th coordinate indicates that edge $e^0_i$ was deleted.  After deleting the $m$ edges, the result is a rooted $n$-leaf tree $\mathcal{T}_{\sigma}$ with a set of transition matrices $\theta_{\sigma}$ and corresponding probability distribution $P_{\mathcal {T}_{\sigma}, \theta_{\sigma}} \in \mathcal M_{\mathcal T}$ on the leaf states. We can then define a distribution on $n$-tuples of DNA bases from the network as follows
 
 $$
P_{\mathcal N, \theta} = 
\displaystyle \sum_{\sigma \in \{0,1\}^m}
( \prod_{i=1}^m \delta_i ^{1-\sigma_i}(1-\delta_i)^{\sigma_i} ) P_{\mathcal {T}_{\sigma}, \theta_{\sigma}} .
$$

Notice, that while the phylogenetic network model is a mixture model,
 it is not simply a $2^m$-tree phylogenetic mixture model.
This is because in a phylogenetic mixture model, the entries of the transition matrices 
are chosen independently for each of the trees in the mixture. 
However, in the network model, the transition matrix parameters are chosen for the network edges and then inherited by the trees embedded in the network as pictured in Figure \ref{fig: TreeMixture}.
It is still the case that the network model has a polynomial parameterization from the stochastic parameter space of the network to the probability simplex. For example, in the case we are considering in this paper, where $m=1$, we can denote the tree obtained by deleting 
$e^0_{1}$ by $\mathcal{T}_1$ and 
the tree obtained by deleting 
$e^1_{1}$ by $\mathcal{T}_2$.
Then the model of the network $\mathcal{M}_\mathcal{N}$ is the image of the polynomial map from the parameter space of the network to the probability simplex,
$$ 
\psi_\mathcal{N}:
\Theta_\mathcal{N} \times [0,1] 
\to
\Delta^{4^n - 1}, \text{  }
(\theta,\delta) \mapsto 
\delta\psi_{\mathcal{T}_1}(\theta) + 
( 1 - \delta)\psi_{\mathcal{T}_2}(\theta) .$$
It is also worth noting that
$\mathcal{T}_1$ and $\mathcal{T}_2$ may have the same
topology but where the network parameters associated to the edges are different.
This is the case with the 3-cycle network depicted in Figure \ref{fig: TreeMixture}.

\begin{ex}
\begin{figure}[h]
	\caption{A 4-leaf 3-cycle network $\mathcal N$ and the two embedded trees $\mathcal T_1$ and $\mathcal T_2$ obtained by deleting reticulation
	edges. }
	\label{fig: TreeMixture} 
		\includegraphics[width=14cm]{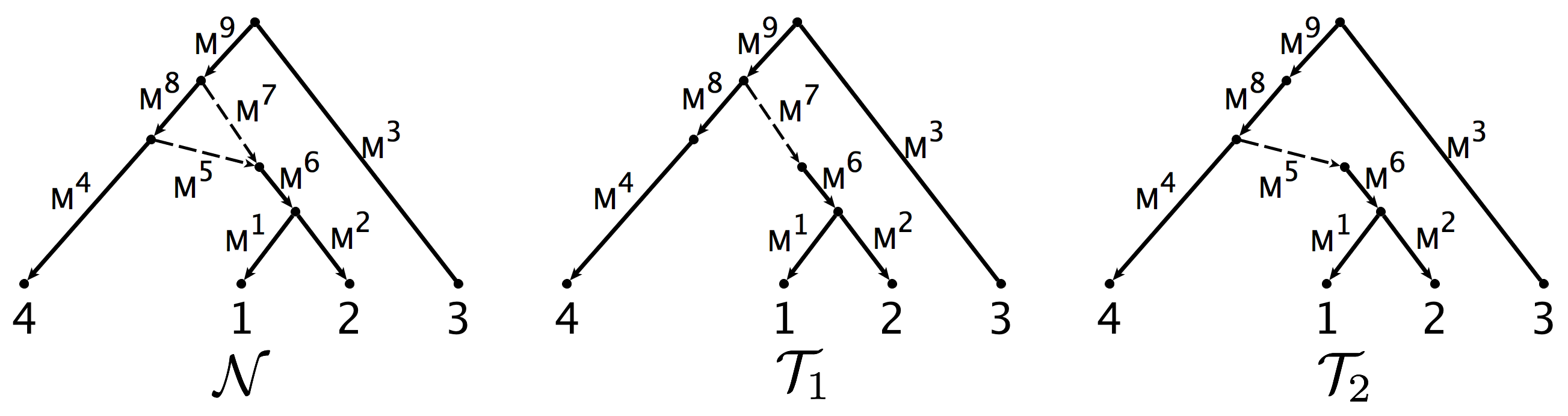}
\end{figure}

The network $\mathcal N$ in Figure \ref{fig: TreeMixture} is a network with two reticulation edges $e_5$ and $e_7$ labeled with the transition matrices $M^{e_5}$ and $M^{e_7}$ respectively (in Figure \ref{fig: TreeMixture}, the $e$'s in the superscripts are suppressed for aesthetics).  We delete $e_5$ and keep $e_7$ with probability $\delta$, and delete $e_7$ and keep $e_5$ with probability $1 -\delta$.  These two possibilities give rise to the two trees $\mathcal T_1$ and $\mathcal T_2$ in Figure \ref{fig: TreeMixture}.  Notice how the transition matrices on $\mathcal T_1$ and $\mathcal T_2$ are inherited from $\mathcal N$.  Thus, we can view the Markov model on $\mathcal N$ as a 2-tree mixture model with additional algebraic relationships among the stochastic parameters.
\end{ex}

The description we have given above works for any nucleotide substitution model. However, in a time-reversible model, the location of the root is unidentifiable \cite{felsenstein1981}. Therefore, for a time-reversible model, we obtain the same distribution by computing each $P_{\mathcal {T}_{\sigma}, \theta_{\sigma}}$ after
unrooting the tree $\mathcal {T}_{\sigma}$. In fact, we obtain the same distribution as from $\mathcal{N}$ if we instead define the model on the
semi-directed network topology of $\mathcal{N}$. 
This implies that for a time-reversible model, any two networks that share the same semi-directed network topology, such as the two networks pictured in Figure \ref{fig: sametopo}, will yield the same distribution.

\begin{figure}
\caption{Two phylogenetic networks with the same semi-directed topology.}
\label{fig: sametopo}
\includegraphics[width=10cm]{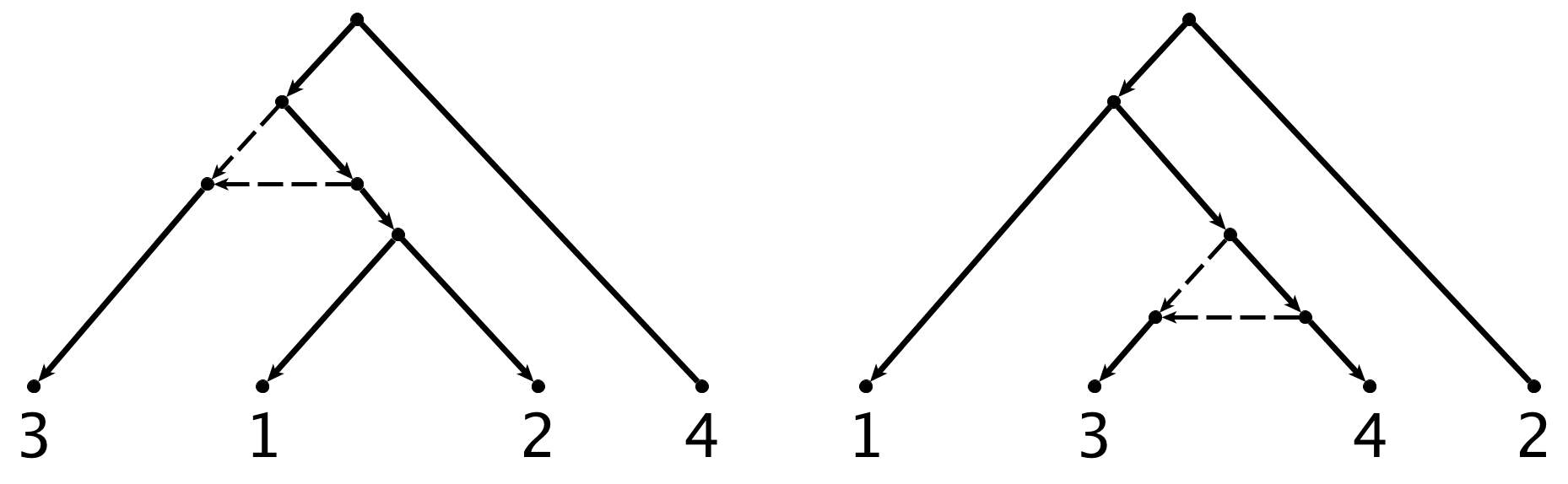}
\end{figure}

\section{Algebraic statistics and generic identifiability}
\label{sec: algstats}

One of the insights of algebraic statistics is that many properties of phylogenetic models can be determined by ignoring the stochastic restrictions on the parameters and regarding $\psi_\mathcal{T}$ as a complex polynomial map. Thus, to answer many questions, it is often enough to consider only the Zariski closure of the image of $\psi_\mathcal{T}$. This is a complex algebraic variety which we denote $\mathcal{V}_\mathcal{T}$. Likewise, in this paper, we will work with
the Zariski closure of the image of $\psi_\mathcal{N}$, the algebraic variety $\mathcal{V}_\mathcal{N}$. Once we have made this change we refer to the formerly stochastic parameters of the model as the \emph{numerical parameters} to distinguish them from the network parameter.
Assuming there are $s$ stochastic parameters, 
we slightly abuse notation and write this new map as $\psi_\mathcal{N}: \mathbb{C}^s \to \mathbb{C}^{4^n}$.

An important question about any model is whether or not the parameters of the model are identifiable. 
For phylogenetic network models, the identifiability of the underlying network parameter is particularly important. 
If we are able to find a network and a choice of stochastic parameters that yield a distribution 
that matches our data, we would like to infer the history of the taxa under consideration from the network topology. 
To do this, we must ensure that the network we have found is the only such network for which it is
possible to do so. More formally, the network topology of an $n$-leaf network model is identifiable 
if given any two $n$-leaf networks $\mathcal{N}_1$ and $\mathcal{N}_2$, the intersection of their models
$\mathcal{M}_{\mathcal{N}_1} \cap \mathcal{M}_{\mathcal{N}_2}$ is empty. This notion of identifiability tends to be too 
strong in practice, and instead, it is often only possible to prove \emph{generic identifiability}. 

\begin{defn} 
\label{defn: genericallyidentifiable}
The network parameter of a phylogenetic network model is \emph{generically identifiable} if given any two 
$n$-leaf networks, $\mathcal{N}_1$ and $\mathcal{N}_2$, the set of parameters in $\Theta_{\mathcal{N}_1}$ that 
$\psi_{\mathcal{N}_1}$ maps into $\mathcal{M}_{\mathcal{N}_2}$ is a set of Lebesgue measure zero.
\end{defn}

In other words, the network parameter is generically identifiable, if, given a specific network, the distribution obtained from a generic choice of stochastic parameters could have only come from this network. 
 To prove the generic identifiability of the network parameter of a phylogenetic network model,
 we will need to be able to distinguish networks. 
 
 \begin{defn}
\label{defn: indistinguishable}
Two distinct $n$-leaf networks
$\mathcal{N}_1$ and
$\mathcal{N}_2$
 are \emph{distinguishable} if 
$\mathcal{V}_{\mathcal{N}_1} \cap \mathcal{V}_{\mathcal{N}_2}$ 
is a proper subvariety of $\mathcal{V}_{\mathcal{N}_1} $
and of $\mathcal{V}_{\mathcal{N}_2} $. 
Otherwise, they are \emph{indistinguishable}.
\end{defn}

Though distinguishability is phrased in terms of varieties, it will 
often be easier to work with the vanishing ideal of the network $\mathcal{N}$.
The vanishing ideal $I_{\mathcal{N}}$  is the set of polynomials that evaluate to zero everywhere on 
$\mathcal{V}_\mathcal{N}$ (or equivalently $\mathcal{M}_\mathcal{N}$). 
Ideals for $n$-leaf network models are contained in polynomial rings where the indeterminates are
 indexed by $n$-tuples of the DNA bases. 
 That is, for a fixed choice of model and an $n$-leaf network
 $\mathcal{N}$,
 $$I_{\mathcal{N}} \subseteq R_{n} := \cc[p_{i_1\ldots i_n} : (i_1, \ldots, i_n) \in \{A,C,G,T\}^n].$$
The elements of the ideal of a phylogenetic model are referred to as \emph{phylogenetic invariants},
and they have played a key role in proving identifiability results for several phylogenetic models
(e.g., \cite{Allman, Long2015a, chifmankubatko2015, Rhodes2012, Long2017}).
The following proposition shows the connection between generic identifiability 
and distinguishing networks.

\begin{op} 
\label{prop: identifiability}
The network parameter of a phylogenetic network model is generically identifiable if 
for all $n \in \mathbb{N}$, all 
pairs of $n$-leaf networks are distinguishable. 
\end{op}

\begin{proof}
Let $\mathcal{N}_1$ and $\mathcal{N}_2$ be distinguishable $n$-leaf networks. 
By definition, this means that $\mathcal{V}_{\mathcal{N}_1} \cap \mathcal{V}_{\mathcal{N}_2}$ 
is a proper subvariety of $\mathcal V_{\mathcal{N}_2} $.
This implies that there exists $f \in I_{\mathcal{N}_1}$ such that $f \not \in I_{\mathcal{N}_2}$ and so that
$$f \circ \psi_{\mathcal{N}_2}: \mathbb{C}^{s_2} \to \mathbb{C}$$
is not identically zero. 
Therefore, the set of stochastic parameters in $\Theta_{\mathcal{N}_2}$ mapping into 
$\mathcal{M}_{\mathcal{N}_1} \cap \mathcal{M}_{\mathcal{N}_2}$ is contained in the proper algebraic subvariety
$$V(f \circ \psi_{\mathcal{N}_2}) := \{\theta \in \mathbb{C}^{s_2}: (f \circ \psi_{\mathcal{N}_2})(\theta) = 0\} \subsetneq \mathbb{C}^{s_2}.$$
This implies that the set of stochastic parameters in $\Theta_{\mathcal{N}_2}$ mapping into 
$\mathcal{M}_{\mathcal{N}_1} \cap \mathcal{M}_{\mathcal{N}_2}$ must be measure zero inside of 
 $\Theta_{\mathcal{N}_2}$.
Otherwise, $V(f \circ \psi_{\mathcal{N}_2})$ would include all of $\mathbb{R}^{s_2}$, and since the real numbers are Zariski dense, it would include all of $\mathbb{C}^{s_2}$, a contradiction.
\end{proof}

\begin{rmk}
While we define generic identifiability of the network parameter to be a condition on \emph{all} pairs of network models from the class of $n$-leaf network models, we could easily modify the definition to be a condition on all pairs of network models from a \emph{subclass} of the $n$-leaf network models. Furthermore, we can modify Proposition \ref{prop: identifiability} to show the identifiability
of the network parameter in a phylogenetic network model on a subclass of network models
by showing all network models in the subclass are distinguishable. \end{rmk}

The preceding remark will be important, since, as we will see in Section \ref{sec: 4leafnetworks}, 
arbitrary Jukes-Cantor networks are not distinguishable.
We will even see that the same is true for the class of all cycle-networks, and we will have
to restrict to the subclass of large-cycle networks to find a class
of network models for which the semi-directed 
network topology is generically identifiable.

\subsection{The Fourier-Hadamard Transform}
\label{sec: fourier}

Our approach for proving many of the results in this paper
will be to use some of the computational algebra techniques outlined above. 
First, we will perform a linear change of coordinates called the 
Fourier-Hadamard transform \cite{Evans1993, Szekely1993} that
will make the parameterizations of the tree-based phylogenetic models monomial. 
This means that the cycle-network models will be parameterized by binomials, 
greatly reducing the computational time. 
Working in the transformed coordinates is common in phylogenetics, including applications involving mixture models \cite{Allman, Long2015a}. 
Since all of the computations referenced in the next section will be performed in Fourier coordinates,
we provide here a basic explanation of the Fourier parameterization. More details can be found in \cite{Evans1993, Szekely1993, Sturmfels2005}.

The Fourier-Hadamard transform applies to a particular 
class of phylogenetic models called \emph{group-based models}.

\begin{defn}
 A phylogenetic model is \emph{group-based} if there exists a group $\mathcal G$, a map $L: \{A,C,G,T\}  \rightarrow \mathcal G$, and functions $f_{e} : \mathcal G \rightarrow \rr $ associated to the edges of $\mathcal{T}$ such that, for each edge $e$ of $\mathcal{T}$,  $M^e_{ij} = f_{e}(L(j)-L(i)).$
 \end{defn}
 
Our definition here is less general than is usually given and applies only to 4-state models of DNA evolution.
However, the definition encompasses many commonly used models in phylogenetic applications including the Jukes-Cantor, the Kimura 2-parameter, and the Kimura 3-parameter models.
The transformation applies to both the parameter space and the space of probability coordinates. 
To write the new parameterization, we let $\Sigma(\mathcal{T})$ represent the set of 
edges of $\mathcal{T}$, and, for an edge $e \in \Sigma(\mathcal{T})$, we write $B_e|B_e'$ for the split of the leaf labels
induced by removing $e$ from $\mathcal{T}$. We also use the symbols $\{A,C,G,T\}$ as shorthand for the group elements
$\{L(A),L(C),L(G),L(T)\}$. The $|\Sigma(\mathcal{T})|\cdot|\mathcal G|$ transformed parameters are written as
 $a^{e}_g$, and the new $|\mathcal G|^n$ coordinates of the image space are parameterized as follows

 \begin{displaymath}
   q_{g_1 \ldots g_n}  = \left\{
     \begin{array}{cl}
       \displaystyle \prod_{{e} \in \Sigma(T)} a^{e}_{\sum_{j \in B_e} g_j} & \text{if }\displaystyle \sum_{j = 1}^{n} g_j= 0 \\
       0 & \text{otherwise.}
     \end{array}
   \right.
\end{displaymath} 

Importantly, as was noted in \cite{Allman}, the linearity of the transform means that it applies to mixtures 
of tree-based phylogenetic models as well. As we have been careful to point out, the network models
studied here are not the same as arbitrary $2$-tree mixtures. 
However, we can obtain the transformed parameterization
of a cycle-network model by identifying parameters 
in the Fourier parameterization of a $2$-tree mixture.

For the Jukes-Cantor model, the group $\mathcal G$ is $\mathbb{Z}_2 \times \mathbb{Z}_2$ and
we arbitrarily set $L(A) = 0$
and then insist that for each edge $e$ of $\mathcal{N}$, $f^e(C) = f^e(G) = f^e(T)$. 
This implies that $a^e_C = a^e_G = a^e_T$, and 
the stochastic condition
in the probability space forces 
$a^e_A =1$. 
We will ignore this last condition, which effectively homogenizes the parameterization
and allows us to work projectively.
The example below shows how to obtain the parameterization for a 4-leaf 3-cycle network.

\begin{ex} Shown below is the parameterization of a Fourier coordinate
of the 4-leaf 3-cycle network pictured in Figure \ref{fig: FourierExample}.
To simplify the notation we use $a^i_g$ for the parameters rather than $a^{e_i}_g$. The first term
of the parameterization is the parameterization of the tree induced by removing the reticulation edge
$e_5$ and the second from removing the edge $e_7$. 
\begin{align*}
q_{ACGT} &= \delta_1(a^{1}_{A}a^{2}_{C}a^{3}_{G}a^{4}_{T}a^{6}_{A + C}a^{7}_{A + C}a^{8}_{T}) + 
\delta_2(a^{1}_{A}a^{2}_{C}a^{3}_{G}a^{4}_{T}a^{5}_{A + C}a^{6}_{A + C}a^{8}_{G}) \\
&= \delta_1(a^{1}_{A}a^{2}_{C}a^{3}_{C}a^{4}_{C}a^{6}_{C}a^{7}_{C}a^{8}_{C}) + 
\delta_2(a^{1}_{A}a^{2}_{C}a^{3}_{C}a^{4}_{C}a^{5}_{C}a^{6}_{C}a^{8}_{C}). 
\end{align*}
Notice, we can reparameterize this variety by replacing
$\delta_1a^7_g$ with $a^7_g$ and $\delta_2a^5_g$ with $a^5_g$.
Thus, we can write

$$q_{ACGT} =a^{1}_{A}a^{2}_{C}a^{3}_{C}a^{4}_{C}a^{6}_{C}a^{7}_{C}a^{8}_{C} + 
a^{1}_{A}a^{2}_{C}a^{3}_{C}a^{4}_{C}a^{5}_{C}a^{6}_{C}a^{8}_{C}.$$

\end{ex}

\begin{figure}[h]
	\caption{A 4-leaf $3$-cycle network.}
	\label{fig: FourierExample} 
		\includegraphics[width=4cm]{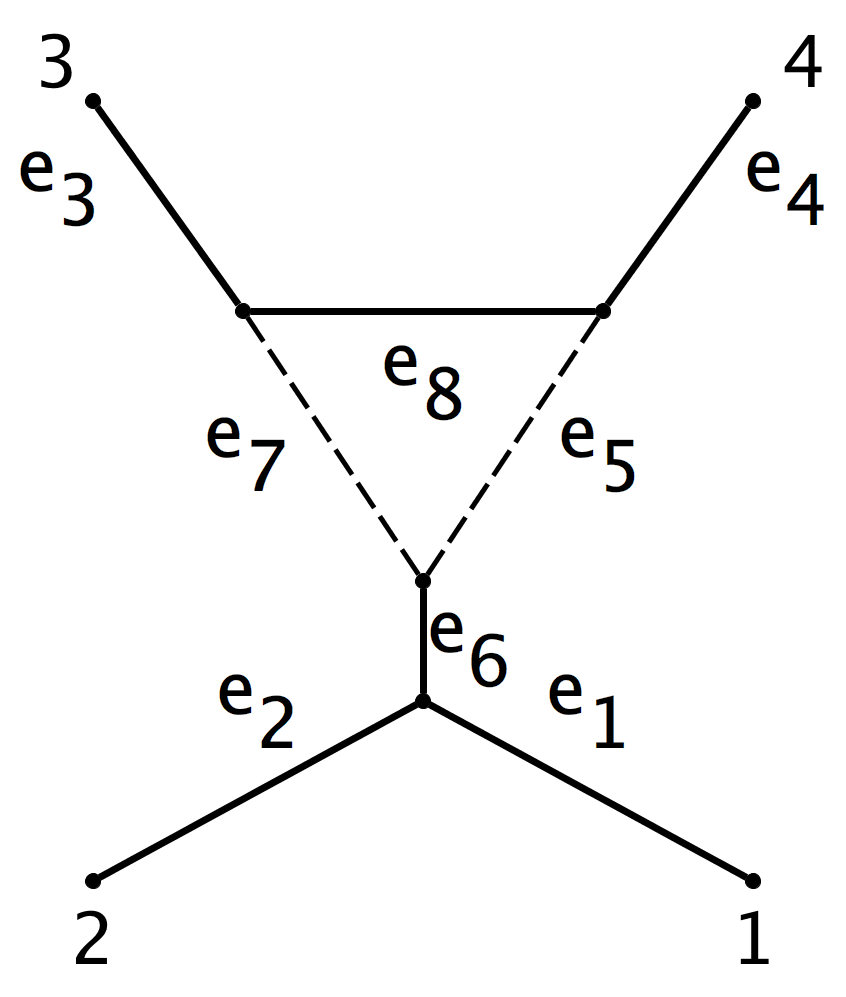}
\end{figure}

\section{Identifiability of cycle-networks}
\label{sec: results}

One of the techniques that has proven successful for establishing 
identifiability for phylogenetic mixture models is to first establish the result for mixtures
on trees with a few leaves. The idea is then to show that distinct mixtures on 
trees with many leaves can always be restricted to a subset of the leaves
on which they remain distinct. 
Our idea here is essentially the same: first prove some
identifiability results for cycle-networks with few leaves and then show
 that these imply identifiability for cycle-networks with any number of leaves. 
To begin, we introduce the concept
of restricting a phylogenetic network.

 \begin{defn}
 \label{defn: Restriction}
Let $\mathcal{N}$ be an $n$-leaf phylogenetic network with root $\rho$,
and let $A \subseteq [n]$. The \emph{restriction of $\mathcal{N}$
to $A$} is the phylogenetic network $\mathcal{N}_{|A}$ constructed by
\begin{enumerate}[(i)]
\item Taking the union of all directed paths from 
$\rho$ to a leaf labeled by an element of $A$.
\item Deleting all vertices that lie above the last such vertex on all 
paths.
\item Suppressing all degree two vertices other than the root.
\item Removing all parallel edges.
\item Applying steps (iii) and (iv) until the network is a phylogenetic network.
\end{enumerate}
\end{defn}

The network constructed after step (i) of Definition \ref{defn: Restriction}
defines a network model which is the same as the phylogenetic
network model $\mathcal{M}_{\mathcal{N}_{|A}}$. 
To see this, notice that it does not alter the model
to delete pendant edges above the vertex 
described in (ii) and reroot at this vertex.
This is because given any previous root distribution $\boldsymbol{\pi}$
and pendant edge $e$, we can simply remove $e$ and 
choose the new root distribution in the model
to be $\boldsymbol{\pi}M^e$. Likewise, any non-root degree two 
vertex has two edges incident to it. Those edges can be replaced
by a single edge with transition matrix that is the product of the 
transition matrices of the two incident edges. If either of the 
incident edges was a reticulation edge, then the new edge is
also a reticulation edge and keeps the same reticulation edge parameter.

Finally, if there are two parallel edges $e^i_0$ and $e^i_1$, they must
be reticulation edges  
and can be replaced by a single edge
with transition matrix $\delta^i M^{e^i_0} + (1 - \delta_i) M^{e^i_1}$.
Thus, it is clear that the phylogenetic network model on
the network constructed after step (i) of Definition \ref{defn: Restriction}
is contained in $\mathcal{M}_{\mathcal{N}_{|A}}$. 
The other containment is easily realized by setting some
of the transition matrices to the identity in the 
network constructed after step (i) of Definition \ref{defn: Restriction}.
The utility of the restriction operation comes from the following proposition. 

\begin{op} 
\label{prop: restricting}
Let $\mathcal{N}$ be an $n$-leaf phylogenetic network and $A \subseteq [n]$. Then
$\mathcal{M}_{\mathcal{N}_{|A}}$ is the image of the model 
$\mathcal{M}_{\mathcal{N}}$ under the 
marginalization map
$\mu_A: \mathbb{C}^{4^n} \to \mathbb{C}^{4^{|A|}}$
defined by marginalizing over all the states of the
 leaves labeled by elements of $[n] \setminus A$.
\end{op}

\begin{proof}
Just as described for a tree in Section \ref{sec: tree models}, given an assignment of states to the 
vertices of $\mathcal{N}$, we can compute the probability of observing this state using the
root distribution, the transition matrices, and the reticulation edge parameters, $\delta_i$. 
Let $\theta$ be a choice of parameters for the model on $\mathcal{N}$.
The distribution 
$P_{\mathcal{N}, \theta} \in \mathcal{M}_\mathcal{N}$ 
can then be computed by marginalizing over all the states
of the non-leaf vertices. The network constructed after
step (i) of Definition \ref{defn: Restriction} defines a distribution which is computed
by further marginalizing over all states of the leaves not labeled by elements of $A$.
As we argued above, this distribution is contained in $\mathcal{M}_{\mathcal{N}_{|A}}$ and
is precisely the image of $P_{\mathcal{N}, \theta}$ under $\mu_A$. Therefore 
$\text{Im}(\mu_A) \subseteq \mathcal{M}_{\mathcal{N}_{|A}}.$

Choosing any parameters $\theta_A$ yielding a distribution
 $P_{\mathcal{N}_{|A},\theta_A} \in \mathcal{M}_{\mathcal{N}_{|A}}$, 
we can choose matching parameters for the edges shared by 
$\mathcal{N}$
and 
$\mathcal{N}_{|A}$ 
and extend this with any choice of parameters 
$\theta'_A$ to the
rest of $\mathcal{N}$. 
Then 
$\mu_A(P_{\mathcal{N}, \theta'_A}) = P_{\mathcal{N}_{|A},\theta_A},$
which implies that  
$\text{Im}(\mu_A) = \mathcal{M}_{\mathcal{N}_{|A}}.$
\end{proof}

In Section \ref{sec: algstats} we argued that to prove identifiability for a class of $n$-leaf
networks, it is enough to show that any two networks in the class are distinguishable. 
One of the nice applications of Proposition \ref{prop: restricting} is the following result.

\begin{op}
\label{prop: restriction distinguishable}
Let $\mathcal{N}_1$ and $\mathcal{N}_2$ be distinct $n$-leaf networks and let $A \subseteq [n]$.
If $\mathcal{N}_{1|A}$ and $\mathcal{N}_{2|A}$ are distinguishable, then so are
$\mathcal{N}_1$ and $\mathcal{N}_2$.
\end{op} 

\begin{proof}
Suppose $\mathcal{N}_{1|A}$ and $\mathcal{N}_{2|A}$ are distinguishable. 
Without loss of generality, we will show
$\mathcal{V}_{\mathcal{N}_1} \cap \mathcal{V}_{\mathcal{N}_2}$ is a proper
subvariety of $\mathcal{V}_{\mathcal{N}_1}$. 
By definition of distinguishable, we have
$\mathcal V_{\mathcal{N}_{1|A}} \cap \mathcal V_{\mathcal{N}_{2|A}} \subsetneq \mathcal V_{\mathcal{N}_{1|A}}$.
Therefore, there exists $f_2 \in R_{|A|}$ that vanishes on $\mathcal{M}_{\mathcal{N}_{2|A}}$
but not on $\mathcal{M}_{\mathcal{N}_{1|A}}$. Letting $\phi_A: R_{|A|} \to R_n $  be the
ring homomorphism corresponding to $\mu_A$, then the polynomial $\phi_A(f_2) \in R_n$ vanishes
on $\mathcal{M}_{\mathcal{N}_{2}}$ but not on $\mathcal{M}_{\mathcal{N}_{1}}$.
Therefore, 
$\mathcal{V}_{\mathcal{N}_1} \cap \mathcal{V}_{\mathcal{N}_2}$ is a proper
subvariety of $\mathcal{V}_{\mathcal{N}_1}$.
\end{proof}

It is possible that after restricting, or indeed, even after unrooting,
that a cycle-network becomes a 2-cycle network. Such a
network must necessarily have parallel edges, which, as
pointed out in the discussion preceding Proposition \ref{prop: restricting},
can be suppressed without altering the network model.
This implies immediately that the models for 2-cycle networks 
are phylogenetic tree models. 
This suggests that perhaps we should exclude 2-cycle networks
entirely from the class of networks considered 
in order to preserve a statement about the generic identifiability of the network
topology. However, even this will not be enough.

\begin{op}
\label{prop: notidentifiable}
For the CFN, JC, K2P, and K3P models, the ideal for any 3-leaf 3-cycle network is the zero ideal.
\end{op}

\begin{proof} 
Up to relabeling, there is only one 3-leaf 3-cycle network topology 
(Figure \ref{fig: 3LeafTopologies}), and by computation, we can verify the ideal of this network is trivial for CFN, JC, K2P, and K3P. 
\end{proof}

\begin{rmk}  For the computations in Proposition \ref{prop: notidentifiable} and all other computations referenced in this paper, we work modulo the set of linear invariants that hold for \emph{every} $n$-leaf network for the group-based model specified.  All computations are performed in Macaulay2 \cite{M2} and are available in the supplementary materials available on the authors' websites.
\end{rmk}

\begin{figure}[h]
	  \caption{The single non-tree 3-leaf network topology.}
	\label{fig: 3LeafTopologies} 
	\begin{center}
		\includegraphics[width=3cm]{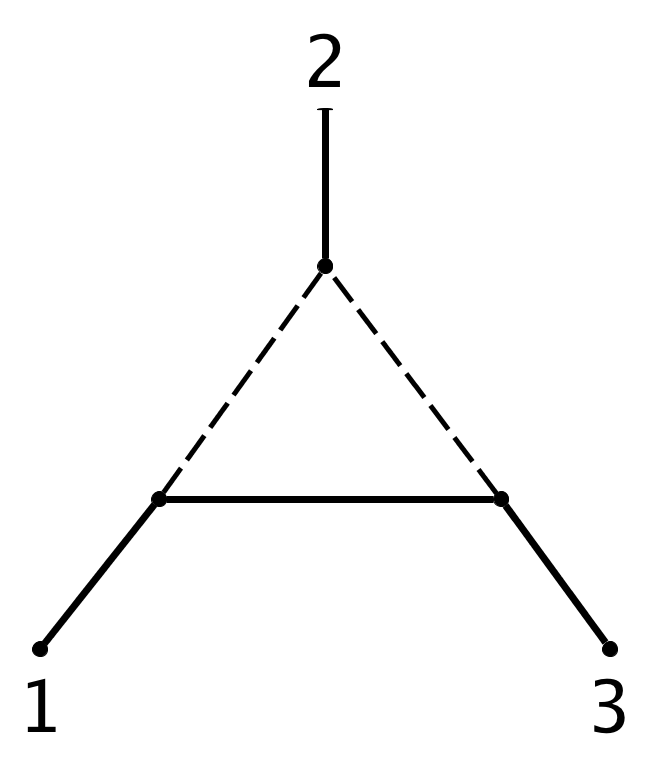}
		\end{center}
\end{figure}

Therefore, in order to find a class of network models for which we can establish indentifiability results, we need to start by considering at least 4-leaf networks.
For the 3-leaf 3-cycle network, if we restrict the parameter space by setting each of the leaf transition matrices to the identity matrix, the corresponding variety still fills the ambient space.
This fact will later prove important when investigating which
4-leaf network ideals are contained in one another.

\subsection{Distinguishing 4-leaf Networks}
\label{sec: 4leafnetworks}

After unrooting and removing parallel edges, there are,
up to relabeling, only four semi-directed 4-leaf cycle-network topologies. 
One of these is the 4-leaf unrooted tree itself, the other
three are pictured in Figure \ref{fig: 4LeafTopologies}.

\begin{figure}[h]
	\caption{The three non-tree 4-leaf cycle-network topologies.}
	\label{fig: 4LeafTopologies} 
	\begin{center}
		\includegraphics[width=10cm]{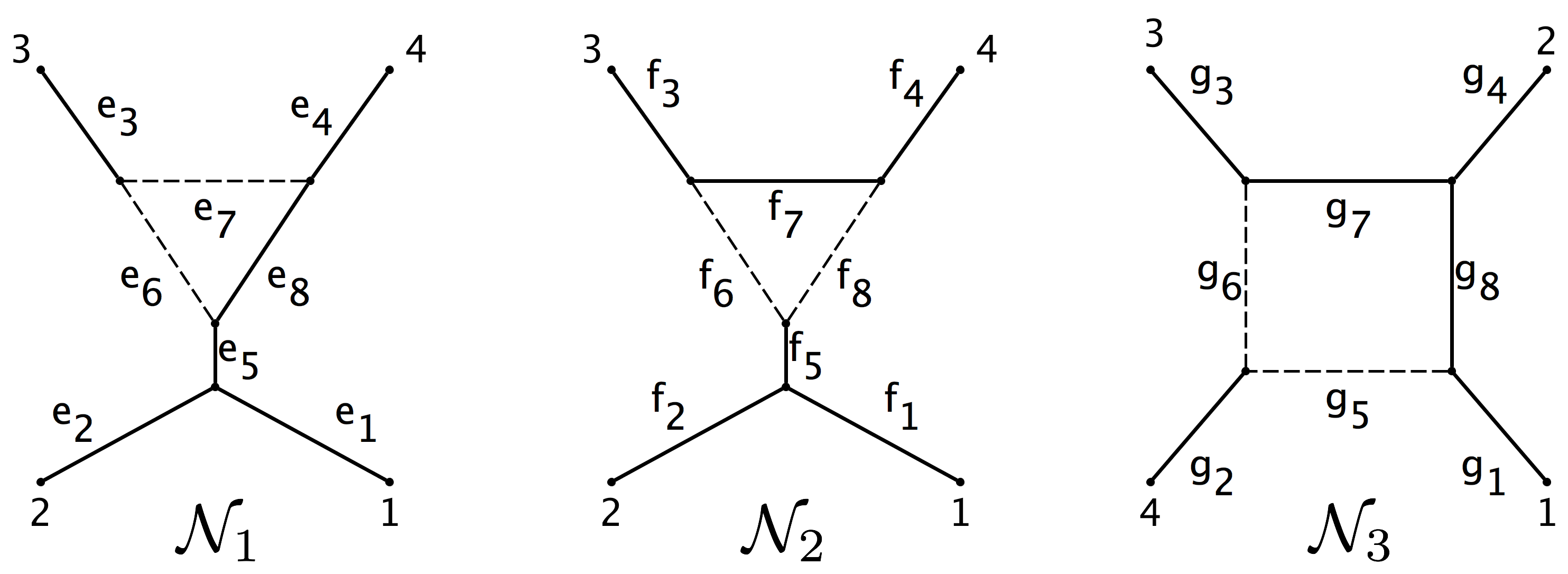}
		\end{center}
\end{figure}

Therefore, up to an action of $S_4$ on the leaf labels, there are at 
most four different 4-leaf cycle-network ideals. In fact, there are
exactly three for the Jukes-Cantor model.

\begin{op} 
\label{prop: sameskeleton}
The Jukes-Cantor network ideals for the two 4-leaf $3$-cycle networks labeled as in Figure \ref{fig: 4LeafTopologies} are equal. 
\end{op}

\begin{proof}
The parameterization for the variety of  $\mathcal{N}_1$ in the Fourier coordinates
is given by 
\begin{align*}
q_{i_1i_2i_3i_4} &= a^1_{i_1}a^2_{i_2}a^3_{i_3}a^4_{i_4}a^5_{i_1 + i_2}a^6_{i_3}a^8_{i_4} + 
a^1_{i_1}a^2_{i_2}a^3_{i_3}a^4_{i_4}a^5_{i_1 + i_2}a^7_{i_3}a^8_{i_1 + i_2} \\
&= a^1_{i_1}a^2_{i_2}(a^3_{i_3}a^4_{i_4}a^5_{i_1 + i_2}a^6_{i_3}a^8_{i_4} + 
a^3_{i_3}a^4_{i_4}a^5_{i_1 + i_2}a^7_{i_3}a^8_{i_1 + i_2}) \\
\end{align*}
for each
$(i_1,i_2,i_3,i_4) \in (\mathbb{Z}_2 \times \mathbb{Z}_2)^4$.
The term in parentheses is exactly the parameterization of
the Fourier coordinate $q_{i_3i_4(i_1 + i_2)}$ for the variety of 
the 3-leaf 3-cycle network
we obtain by pruning off the leaves $1$ and $2$ from $\mathcal{N}_1$. 
Likewise, letting
$b^i_{g}$ be the Fourier parameters for $\mathcal{N}_2$, we have 
$$q_{i_1i_2i_3i_4} = 
b^1_{i_1}b^2_{i_2}(b^3_{i_3}b^4_{i_4}b^5_{i_1 + i_2}b^6_{i_1 + i_2}b^7_{i_4} + 
b^3_{i_3}b^4_{i_4}b^5_{i_1 + i_2}b^7_{i_3}b^8_{i_1 + i_2}).$$
Again, the term in parentheses is the parameterization of
$q_{i_3i_4(i_1 + i_2)}$ for the variety of  
the 3-leaf 3-cycle network we obtain by pruning the leaves 
$1$ and $2$ from $\mathcal{N}_2$. 

Without loss of generality, specify the $a^i_g$ to 
obtain a point in $\mathcal{V}_{\mathcal{N}_1}$.
Since there are no invariants for any 3-leaf 3-cycle network, for a generic
choice of parameters, we can choose the $b^i_g$ for $3 \leq i \leq 8$ so that 
$$
(a^3_{i_3}a^4_{i_4}a^5_{i_1 + i_2}a^6_{i_3}a^8_{i_4} + 
a^3_{i_3}a^4_{i_4}a^5_{i_1 + i_2}a^7_{i_3}a^8_{i_1 + i_2}) = 
(b^3_{i_3}b^4_{i_4}b^5_{i_1 + i_2}b^6_{i_1 + i_2}b^7_{i_4} + 
b^3_{i_3}b^4_{i_4}b^5_{i_1 + i_2}b^7_{i_3}b^8_{i_1 + i_2})$$
for all
$(i_1,i_2,i_3,i_4) \in (\mathbb{Z}_2 \times \mathbb{Z}_2)^4$.
Further choosing $b^1_g = a^1_g$ and $b^2_g = a^2_g$
shows that this point is also in 
$\mathcal{V}_{\mathcal{N}_2}$.
Since a generic choice of parameters for $\mathcal{N}_2$ must also 
map into $\mathcal{V}_{\mathcal{N}_1}$, it must be that
$\mathcal{V}_{\mathcal{N}_1} = \mathcal{V}_{\mathcal{N}_2}$.
\end{proof}

\begin{rmk}
Proposition \ref{prop: sameskeleton} can be proven more succinctly using
the \emph{toric fiber product}. 
The toric fiber product \cite{Sullivant2007} is a procedure that takes two 
homogeneous ideals (not necessarily toric) in rings with a compatible grading 
and produces a new homogeneous ideal.
For phylogenetic tree models, it has been shown that the toric fiber product can be used to 
construct the ideal associated to a phylogenetic tree by ``gluing" together the 
ideals associated to claw trees.
Though we do not develop the full machinery here, the details for cycle-networks
closely parallel the situation described for trees in \cite[Section 3.4]{Sullivant2007}. 
In Proposition \ref{prop: sameskeleton}, $\mathcal{N}_1$ and $\mathcal{N}_2$ can both be 
constructed by gluing a 3-leaf claw tree to the 3-sunlet along the edges $e_5$ and $f_5$
respectively. The proof of Proposition \ref{prop: sameskeleton} 
then follows immediately since both network ideals
are equal to the toric fiber product of the zero ideal 
and the ideal for the 3-leaf claw tree.
\end{rmk}

There are $4! = 24$ ways to label each of the three $4$-leaf cycle-network topologies.
However, many of these labelings result in the same ideal. 
For example, swapping the labels $1$ and $3$ in the $4$-cycle network in 
Figure \ref{fig: 4LeafTopologies} does not change the network.
The following proposition classifies all the ideals associated to 
4-leaf Jukes-Cantor cycle-networks.

\begin{op}
\label{prop: AllJCideals}
For 4-leaf Jukes-Cantor cycle-networks, there are 
\begin{itemize}
\item 3 ideals corresponding to 2-cycle networks (trees) that are 6-dimensional.
\item 6 ideals corresponding to 3-cycle networks that are 7-dimensional.
\item 12 ideals corresponding to 4-cycle networks that are 8-dimensional.
\end{itemize} 
\end{op}

\noindent Notice that this situation is in sharp
contrast to the case of Jukes-Cantor mixture models,
and, indeed, all other group-based mixture models, 
where equality between the
number of leaves implies equality between the ideal dimensions 
\cite{Allman, Long2015a, MRC2Paper}.

The dimension results in Proposition \ref{prop: AllJCideals} are obtained by computing 
the ideals in Macaulay2 \cite{M2}; the computations 
are available in the supplementary materials.
Our approach for each ideal is to first obtain 
a set of elements in the ideals by computing 
the ideal only up to a certain degree.
We then use the rank of the Jacobian matrix to construct
a lower bound on the dimension of the ideal. 
Finally, we verify that the elements found in low
degree generate a prime ideal of the correct dimension, 
and hence, form a generating set.

We include at this point a classification of $2, 3,$ and $4$-cycle network ideals for the CFN model. 
This proposition suggests that it may be difficult or impossible to obtain 
strong generic identifiability results for CFN networks and provides motivation 
for beginning with the Jukes-Cantor model. 

\begin{op}
\label{prop: CFN}
For 4-leaf CFN cycle-networks, there are 
\begin{itemize}
\item 3 ideals corresponding to 2-cycle (trees) and 3-cycle networks that are 6-dimensional.
\item 3 ideals corresponding to 4-cycle networks that are 7-dimensional.
\end{itemize} 
\end{op}

Returning again to Jukes-Cantor networks, we have the following corollaries to Proposition \ref{prop: AllJCideals}.

\begin{cor}
\label{cor: cyclecontainment}
 Let $\mathcal{N}_1$ be a $k_1$-cycle network and $\mathcal{N}_2$ be a $k_2$-cycle network.
If $2 \leq k_1 < k_2 \leq 4$, then $\mathcal{V}_{\mathcal{N}_2} \not \subseteq \mathcal{V}_{\mathcal{N}_1}$ and
$\mathcal{I}_{\mathcal{N}_1} \not \subseteq \mathcal{I}_{\mathcal{N}_2}.$
\end{cor}

\begin{cor}
\label{cor: distinct4cycles}
 Let $\mathcal{N}_1$ and $\mathcal{N}_2$ be distinct $4$-leaf $4$-cycle networks.
Then $\mathcal{V}_{\mathcal{N}_2} \not \subseteq \mathcal{V}_{\mathcal{N}_1}$,
$\mathcal{V}_{\mathcal{N}_1} \not \subseteq \mathcal{V}_{\mathcal{N}_2}$,
$\mathcal{I}_{\mathcal{N}_1} \not \subseteq \mathcal{I}_{\mathcal{N}_2}$,
 and
$\mathcal{I}_{\mathcal{N}_2} \not \subseteq \mathcal{I}_{\mathcal{N}_1}.$
\end{cor}

The network ideals described in Proposition \ref{prop: AllJCideals} of the same dimension differ only by a permutation of the coordinates. Since each network ideal is parameterized, the ideal can be written as the
kernel of a homomorphism, and, consequently, it is prime. If an ideal contains a prime ideal of the same dimension, then the two ideals are equal. Therefore, the network ideals of the same dimension are either equal or distinguishable. 

\begin{figure}[h]
	  \caption{The $4$-leaf cycle-network poset with $\mathcal{N}_1 \prec \mathcal{N}_2$
	  if and only if $\mathcal{V}_{\mathcal{N}_1} \subset \mathcal{V}_{\mathcal{N}_2}$. }
	\label{fig: poset} 
	\begin{center}
		\includegraphics[width=14cm]{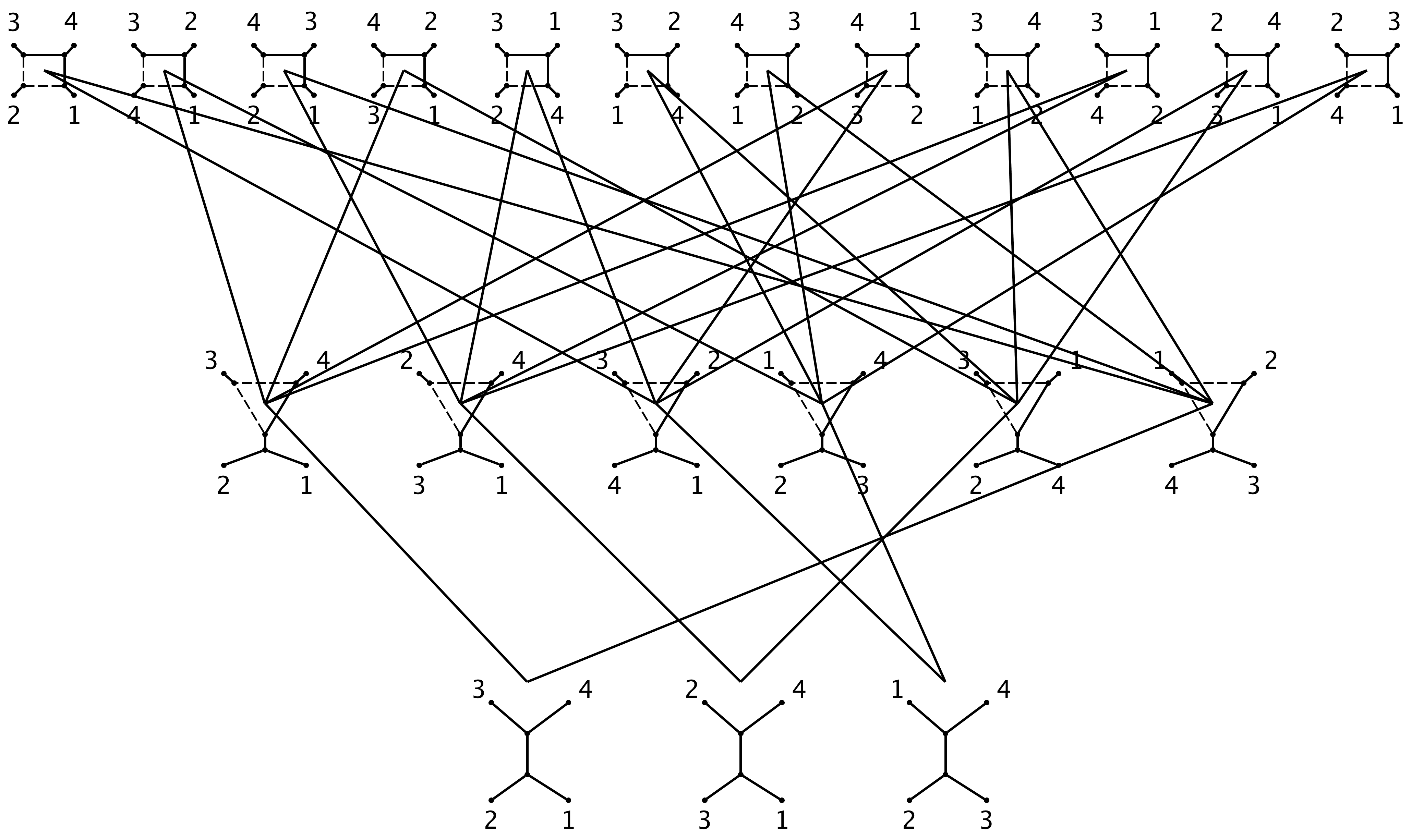}
		\end{center}
\end{figure}

The poset pictured in Figure \ref{fig: poset} shows the containment relationships among the 
21 equivalence classes 
of 4-leaf cycle-network ideals.
It is possible to verify the containment of the 3-cycle network varieties
inside the 4-cycle network varieties by showing the
reverse inclusion of the ideals by computation. 
However, we can also see this from the structure of the networks themselves.

\begin{ex}
\label{ex: collapse edge}
Choose each of the Fourier parameters associated to the edge $g_8$ to be equal to 1
in the 4-cycle network from Figure \ref{fig: 4LeafTopologies}.
This essentially collapses the edge $g_8$ in the network to produce
a 4-leaf 3-cycle network. 
We can construct this new network by attaching a 1-2 cherry to a 
3-leaf 3-cycle network with a single leaf edge removed.
As noted after Proposition \ref{prop: notidentifiable}, 
there are no invariants for the 3-leaf 3-cycle network even with all leaf edges removed.
Therefore, the same arguments from the proof of Proposition \ref{prop: sameskeleton}
show that the variety for this network is equal to both 
$\mathcal{V}_{\mathcal{N}_1}$ and $\mathcal{V}_{\mathcal{N}_2}$.
Therefore, both of these network varieties are contained in $\mathcal{V}_{\mathcal{N}_3}$.
Collapsing the other solid edge $g_7$ in the cycle
of the 4-cycle network shows that $\mathcal{V}_{\mathcal{N}_3}$ also contains the 
variety of any 4-leaf 3-cycle network with a 2-3 cherry.
\end{ex}

\subsection{Distinguishing Large-Cycle Networks}
\label{sec: largekcycles}

In the previous section, we showed that it is possible to distinguish some
cycle-networks with only a few leaves
from one another by computing the ideals for these networks explicitly. 
In this section, we collect
the results needed to prove Theorem \ref{thm: main}. 
That is, we will show that if $\mathcal{N}_1$ is an 
$n$-leaf $k_1$-cycle network and $\mathcal{N}_2$ is a distinct 
$n$-leaf $k_2$-cycle network with $k_1, k_2 \geq 4$, then
$\mathcal{V}_{\mathcal{N}_{1}} \not \subseteq 
\mathcal{V}_{\mathcal{N}_{2}}$. 

The three lemmas below address the three cases, where $k_1 = k_2$, $k_1 > k_2$, and $k_1 < k_2$.
In each of the lemmas we will assume that 
$\mathcal{N}_1$ is a
$k_1$-cycle network and $\mathcal{N}_2$ is a
$k_2$-cycle network.
We assume that the cycle vertices of $\mathcal N_1$ and $\mathcal N_2$ 
are labeled according to the convention described in Section \ref{sec:phylogenetic-networks} 
so that the induced partition of $[n]$ is $A_1 | A_2 | \ldots | A_{k_1}$ in 
$\mathcal N_1$ and $B_1 | B_2 | \ldots| B_{k_2}$ in $\mathcal N_2$. 
The goal in each case of each of the lemmas below will
be to find $S \subset [n]$ such that 
$\mathcal V_{\mathcal N_{1 | S}} \not \subseteq \mathcal V_{\mathcal N_{2 | S}}$,
 which by Proposition \ref{prop: restriction distinguishable}, 
 implies that 
$\mathcal V_{\mathcal N_1} \not \subseteq  \mathcal V_{\mathcal N_2}$.
One result that will use repeatedly is the generic identifiability 
of the tree topology of a Jukes-Cantor tree model.
This is a well-known result with multiple independent proofs \cite{Chang1991, Steel1993}. 

\begin{lemma}\label{lem:cycles-equal}  Let $\mathcal N_1$ and $\mathcal N_2$ be two distinct $k$-cycle networks with $k\geq4$. Then $\mathcal V_{\mathcal N_1} \not \subseteq  \mathcal V_{\mathcal N_2}$ and $\mathcal V_{\mathcal N_2} \not \subseteq \mathcal V_{\mathcal N_1}$.
\end{lemma}

\begin{proof}  

\medskip

\noindent 
\emph{Case 1:} $\{ A_1, \ldots, A_k \} \neq \{B_1, \ldots, B_k\}.$
\medskip

\noindent
Since $\{ A_1, \ldots, A_k \} \neq \{B_1, \ldots, B_k\}$ there exist $\ell, i, j \in [k]$ with $i\neq j$ and $a, b \in [n]$  such that $a, b \in A_{\ell}$ while $a \in B_i$ and $b \in B_j$.  Let $S \subset [n]$ contain $a, b,$ and two additional leaf labels so that
$\mathcal{N}_{2|S}$ is a 4-leaf 4-cycle network.
Since $a,b \in A_{\ell}$, $\mathcal{N}_{1|S}$ is either a 2 or 3-cycle network. 
In either case, Corollary \ref{cor: cyclecontainment} implies that 
$\mathcal V_{\mathcal N_{2 | S}} \not \subseteq \mathcal V_{\mathcal N_{1 | S}}$.
By a similar argument, we can show that 
$\mathcal V_{\mathcal N_{1 | S}} \not \subseteq \mathcal V_{\mathcal N_{2 | S}}$.

\medskip

\noindent 
\emph{Case 2:} $\{ A_1, \ldots, A_k \} = \{B_1, \ldots ,B_k\}$.
\medskip

\noindent
If $\{ A_1, \ldots, A_k \} = \{B_1, \ldots ,B_k\}$, then we can assume that if $A_i =B_j$ then $\mathcal N_{1 |A_i} = \mathcal N_{1 | B_j}$, else the desired result follows from the result for trees.  Furthermore, we can assume that there exists an $i$ such that $A_i \neq B_i$ (else $\mathcal N_1 = \mathcal N_2$). Thus, $B_1, \ldots, B_k$ is simply a reordering of $A_1, \ldots, A_k$, i.e. $A_1| \ldots | A_k = B_{i_1} | \ldots |B_{i_k}$. Without loss of generality, we can view each network as a $k$-sunlet network, with the pendant edges of $\mathcal N_1$ labeled $1, \ldots, k$ starting from the reticulation vertex and proceeding clockwise and with the pendant edges of $\mathcal N_2$ labeled $i_1, \ldots, i_k$ starting from the reticulation vertex and proceeding clockwise. 

Assume $i_1 \neq 1$ and let 
$S = \{1, i_1, a, b\}$ where $a, b \in [k]\setminus \{1, i_1\}$ are two additional leaves with $a \neq b$. 
Then, since they do not have the same reticulation vertex, $\mathcal N_{ 1| S}$ 
and $\mathcal N_{2 | S}$ are distinct 4-leaf 4-cycle networks, and so by Corollary \ref{cor: distinct4cycles},
$\mathcal V_{\mathcal N_{1 | S}} \not \subseteq \mathcal V_{\mathcal N_{2 | S}}$ and 
$\mathcal V_{\mathcal N_{2 | S}} \not \subseteq \mathcal V_{\mathcal N_{1 | S}}$. 
If $i_1 =1$, then choose $S$ to be $[n] \setminus \{1\}$.
Then $\mathcal{N}_{1|S}$ is a caterpillar 
tree with cherries labeled by $\{2,3\}$ and $\{k-1,k\}$ and 
$\mathcal{N}_{2|S}$ is a caterpillar 
tree with cherries labeled by $\{i_2,i_3\}$ and $\{i_{k-1},i_k\}$.
If these trees are not identical, then again, results for single tree models imply
$\mathcal V_{\mathcal N_{1 | S}} \not \subseteq \mathcal V_{\mathcal N_{2 | S}}$ and
$\mathcal V_{\mathcal N_{2 | S}} \not \subseteq \mathcal V_{\mathcal N_{1 | S}}$.
If these trees are identical, then since 
$\mathcal N_{1} \not = \mathcal N_{2}$ it must be that for either $S = \{1,2,3,k\}$
or for $S = \{1,2,k-1,k\}$, $\mathcal N_{1|S}$ and $\mathcal N_{2|S}$
are distinct 4-leaf 4-cycle networks, and the result follows again by 
Corollary \ref{cor: distinct4cycles}.

\end{proof}

\begin{lemma}\label{lem:cycles-greater}  Let $\mathcal N_1$ be a $k_1$-cycle network and $\mathcal N_2$ be a $k_2$-cycle network with $k_1 > k_2$. Then $\mathcal V_{\mathcal N_1} \not \subseteq  \mathcal V_{\mathcal N_2}$.
\end{lemma}

\begin{proof} 
We may assume $k_1 \geq 4$, since otherwise the result follows from Corollary \ref{cor: cyclecontainment}.
Since $k_1 > k_2$, there exist $\ell, i, j$ with $i\neq j$ and $a, b \in [n]$  such that $a, b \in B_{\ell}$ while $a \in A_i$ and $b \in A_j$.  
As in Lemma \ref{lem:cycles-equal} let $S \subset [n]$ contain $a, b$ and two additional leaf labels so that
$\mathcal{N}_{1|S}$ is a 4-leaf 4-cycle network. Again, 
$\mathcal N_{2 | S}$ must be either a tree or a $3$-cycle network and the result follows by
Corollary \ref{cor: cyclecontainment} and Proposition \ref{prop: restriction distinguishable}.
\end{proof}

\begin{figure}
\caption{The two possible networks equal to $\mathcal{N}_{1|S}$ in the proof of 
Lemma \ref{lem:cycles-lesser}.}
\label{fig: 5leaf4cycles}
\includegraphics[width=10cm]{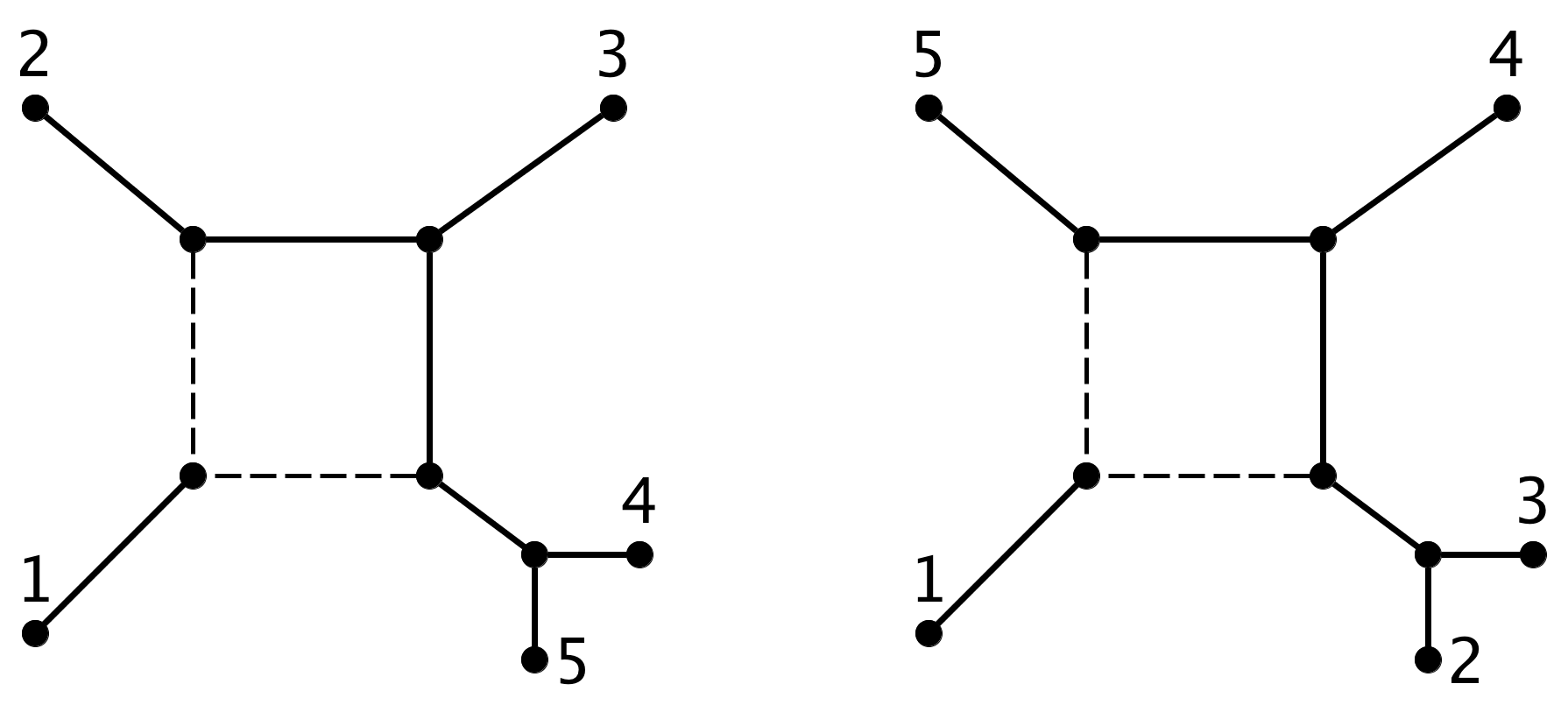}
\end{figure}

\begin{lemma}\label{lem:cycles-lesser}  Let $\mathcal N_1$ be a $k_1$-cycle network and $\mathcal N_2$ be a $k_2$-cycle network with $4 \leq k_1 < k_2$. Then $\mathcal V_{\mathcal N_1} \not \subseteq  \mathcal V_{\mathcal N_2}$.
\end{lemma}

\begin{proof}  
By similar arguments to those in the proof of Lemma \ref{lem:cycles-equal}, if $B_1 | B_2 | \ldots |B_{k_2}$ is not a refinement of $A_1 | A_2 | \ldots | A_{k_1}$, then $\mathcal V_{\mathcal N_1} \not \subseteq  \mathcal V_{\mathcal N_2}$.
Thus, we will assume that $B_1 | B_2 | \ldots |B_{k_2}$ is a refinement of $A_1 | A_2 | \ldots | A_{k_1}$.  

Since $B_1 | B_2 | \ldots |B_{k_2}$ refines $A_1 | A_2 | \ldots | A_{k_1}$, there exist $i,j,\ell_1$ such that $B_i, B_j \subset A_{\ell_1}$. Construct the set $S$ consisting of $a \in B_i$, $b \in B_j$, and any three other leaf labels so
that $\mathcal{N}_1$ is a 5-leaf 4-cycle network. Now by construction, $\mathcal{N}_2$ must be a 5-leaf 5-cycle network. Thus, up to relabeling, we can assume that 
$\mathcal{N}_{2|S}$ is the 5-sunlet network with the leaf labeled by 1 attached to the 
reticulation vertex and all other leaves labeled in consecutive order around the sunlet.
$\mathcal{N}_{1|S}$ could now be any one of several 5-leaf 4-cycle networks.
However, if it is any network other
than one of the two pictured in Figure \ref{fig: 5leaf4cycles}, 
then there exists a 4-element subset $S' \subset S$ 
such that either $\mathcal{N}_{1|S'}$ and $\mathcal{N}_{2|S'}$ are distinct trees or 
$\mathcal{N}_{1|S'}$ is a 3-cycle network and $\mathcal{N}_{2|S'}$ is a tree. 
In either event, this would imply 
$\mathcal{V}_{\mathcal{N}_{1|S'}} \not \subseteq \mathcal{V}_{\mathcal{N}_{2|S'}}$
and the result follows. Thus, we may assume that $\mathcal{N}_{1|S}$
is one of the two networks pictured in Figure \ref{fig: 5leaf4cycles}.

Representing $A$, $C$, $G$, and $T$ by $0$, $1$, $2$, and $3$, the following cubic is in the 5-sunlet network ideal 
\begin{align*}
& {q}_{(0,1,0,1,0)} {q}_{(2,1,1,0,2)} {q}_{(1,1,3,1,2)}-
    {q}_{(0,1,1,0,0)} {q}_{(2,1,0,1,2)}{q}_{(1,1,3,1,2)}+\\
& {q}_{(0,1,1,0,0)} {q}_{(2,1,0,2,1)} {q}_{(1,1,3,1,2)}-
   {q}_{(0,1,0,1,0)} {q}_{(2,1,1,0,2)} {q}_{(2,1,3,1,1)}+\\
& {q}_{(0,1,1,0,0)} {q}_{(2,1,0,1,2)}{q}_{(2,1,3,1,1)}-
   {q}_{(0,1,1,0,0)} {q}_{(2,1,0,2,1)} {q}_{(2,1,3,1,1)}-\\
& {q}_{(0,1,0,1,0)} {q}_{(2,1,1,0,2)} {q}_{(1,1,1,3,2)}+
   {q}_{(0,1,1,0,0)} {q}_{(2,1,0,1,2)} {q}_{(1,1,1,3,2)}-\\
& {q}_{(0,1,1,0,0)} {q}_{(2,1,0,2,1)} {q}_{(1,1,1,3,2)}+
   {q}_{(0,1,1,0,0)} {q}_{(2,1,0,2,1)} {q}_{(2,1,1,1,3)}.
\end{align*}     
Substituting the parameterization for each of the 
two 4-cycle networks pictured into this polynomial, we find that 
it must vanish on $\mathcal{V}_{\mathcal{N}_{2|S}}$ but not on $\mathcal{V}_{\mathcal{N}_{1|S}}$.
Thus, 
$\mathcal{V}_{\mathcal{N}_{1|S}} \not \subseteq
\mathcal{V}_{\mathcal{N}_{2|S}}$.
\end{proof}

Finally, we are able to give the proof of the main theorem.

\begin{proof}[Proof of Theorem \ref{thm: main}]
By Proposition  \ref{prop: identifiability}, 
the network parameter of a phylogenetic network model is generically identifiable if 
for all $n \in \mathbb{N}$, all 
pairs of $n$-leaf networks are distinguishable.  
So let 
$\mathcal{N}_1$ be an $n$-leaf $k_1$-cycle network and 
$\mathcal{N}_2$ be an $n$-leaf $k_2$-cycle network with $k_1,k_2 \geq 4$.
By application of one of 
Lemmas \ref{lem:cycles-equal}, \ref{lem:cycles-greater}, or \ref{lem:cycles-lesser},
$\mathcal{N}_1$ and 
$\mathcal{N}_2$ are distinguishable.
\end{proof}

\section{Discussion and Open Problems}
\label{sec: discussion}

We have shown that the semi-directed network topology of a Jukes-Cantor network is not necessarily 
identifiable even when restricting to networks with a single reticulation vertex. 
In fact, we need to further restrict to the class of large-cycle networks in order for the semi-directed network topology to be generically identifiable.  While this identifiability result covers a large subset of cycle-networks, 
models on networks with small cycles may be of biological interest, and thus, exploration on 
how to use these models effectively is required.    

Furthermore, this paper introduces a collection of algebraic varieties worth deeper investigation. The varieties in this paper are subvarieties of the join varieties associated to 2-tree mixture models, but, as we have seen, the class of network model varieties has different properties than the class of 2-tree mixture varieties. 
Thus, there remain a number of interesting mathematical questions to address.
For example, the results of Lemma \ref{prop: AllJCideals} might suggest that 
for networks with $n$-leaves, the dimension of the network model
increases with cycle size. However, this has not been proven, and indeed, we 
propose the following conjecture to the contrary.

\begin{conj} Let $\mathcal N_1$ and $\mathcal N_2$ be two $n$-leaf large-cycle networks, then 
$\dim ( \mathcal{V}_{\mathcal N_1} ) = \dim (\mathcal{V}_{\mathcal N_2}) $.
\end{conj}

While we are unable to compute the full ideals, the rank of the Jacobian
matrix evaluated at a random point for the 5-leaf $4$-cycle and $5$-sunlet networks is the same, suggesting
that both ideals are the same dimension.
Observe also, that the argument from Example \ref{ex: collapse edge} 
does not apply to 5-sunlet networks. 
This is because when we collapse one of the cycle edges of the
5-sunlet network, the resulting network
is not binary. This was also the case when we collapsed an edge 
in the 4-cycle network in Example \ref{ex: collapse edge}. 
The resulting variety was equivalent to the variety
of a 3-cycle binary network only since the variety of the
3-leaf 3-cycle network with all leaf edges collapsed fills the entire space. 
Thus, it would be interesting to determine if this dimension phenomenon is
isolated to those networks with small cycle size.

The next question we pose comes from the logical step of extending  this work to other models. 

\begin{ques} Is the semi-directed network topology parameter generically identifiable for large-cycle Kimura-2 parameter (Kimura-3 parameter) network models?
\end{ques}

Many of the same techniques should prove fruitful
for these models. The combinatorial arguments used in this paper will apply,
though the number of parameters and the size of the rings
may make the computational steps much more difficult. 
Still, since K2P and K3P are group-based, 
the Fourier transform applies and
finding the necessary invariants is at least within the realm of possibility.
All this of course is only to address the identifiability of the semi-directed
network topology.  
It would also be of practical interest to determine the 
identifiability of the transition matrix parameters.

The identifiability results we obtain in this paper suggest that there could be 
more identifiability issues as we increase the number of reticulation vertices in the
network. For this reason, applying a similar approach to more general classes of networks 
would be of great interest. 
One of the key tools in this paper is Lemma \ref{prop: restriction distinguishable}, which allows us to 
prove the identifiability of networks with any number of leaves
by considering only networks with fewer than five leaves in Section \ref{sec: largekcycles}. 
However, it has already been shown that arbitrary networks cannot be identified
by their subnetworks 
\cite{Huber2015}, 
so this approach has little hope of succeeding in that case.
Instead, the next step might be to examine 
{tree-child} networks which are identifiable from their \emph{trinets}, induced subnetworks
on three leaves \cite{VanIersel2014}. There are some subtleties involved here as well, as the result for trinets applies
to the rooted network topology and we have already seen that there will be indistinguishable
3-leaf {tree-child} networks. Still, it may be possible to make similar arguments by restricting
to $n$-leaf subnetworks for some fixed $n$. Finally, it may be worthwhile to start by examining
slightly more general classes of networks, such as {level-1 networks} with only two 
reticulation vertices.

\section{Acknowledgements}

We would like to thank Seth Sullivant for his insights regarding the toric fiber product and cycle-network ideals.
Colby Long is supported by the Mathematical Biosciences Institute and the 
National Science Foundation under grant DMS-1440386. 
Elizabeth Gross is supported by the National Science Foundation under grant DMS-1620109.

\bibliography{references}
\bibliographystyle{plain}

\end{document}